\documentclass[12pt]{article}

\usepackage{amsmath,amssymb,amsfonts}
\usepackage{amsthm}
\usepackage{booktabs}
\usepackage{multirow}
\usepackage{array}
\usepackage{url}
\usepackage{graphicx}
\usepackage{xcolor}
\usepackage{float}
\usepackage{placeins}
\usepackage{geometry}
\usepackage{hyperref}
\usepackage{adjustbox}

\newtheorem{theorem}{Theorem}
\newtheorem{lemma}{Lemma}
\newtheorem{proposition}{Proposition}
\newtheorem{definition}{Definition}
\newtheorem{remark}{Remark}
\newtheorem{example}{Example}

\newtheorem{corollary}{Corollary}

\newcommand{\Z}{\mathbb{Z}}

\newcommand{\safeincludegraphics}[2][]{%
  \IfFileExists{#2}{\includegraphics[#1]{#2}}{%
    \fbox{\parbox{0.88\linewidth}{\centering Missing figure file: \texttt{\detokenize{#2}}}}%
  }%
}

\title{Local Fault Repair of Perfect Resource Placements in Dense Gaussian Networks}

\author{Bader A. Albader\\
\small Department of Computer Science, Faculty of Science, Kuwait University, Kuwait\\
\small \texttt{albader@cs.ku.edu.kw}}

\date{}

\begin{document}

\maketitle

\begin{abstract}
Perfect resource placement in dense Gaussian networks partitions the network into Lee balls centered at resource nodes. The fault-free placement problem is already classified; this paper studies the complementary post-deployment problem of repairing such placements after resource faults. The paper gives exact local repair theorems for the dense Gaussian placement generated by $t+(t+1)i$; by conjugation and rotation symmetry, the same results hold for the companion generator $(t+1)+ti$. For one failed resource, we prove failure-cell locality, derive the exact replacement number $\rho_G(1)=3$ and $\rho_G(t)=2$ for all $t\ge2$, and prove the sharp minimum-overlap formula $\Omega_G(t)=t+1$ among minimum-size repairs. The overlap lower bound is proved from the corner structure of equal-size Lee balls in the rotated coordinates $u=x+y$ and $v=x-y$, where Gaussian Lee balls become parity-constrained squares. For two failed resources, we prove exact additivity: every pair of failed resource cells requires exactly four local replacements for $t\ge2$, and four always suffice. The two-fault lower bound reduces all relevant resource displacements to two canonical neighboring cases and exhibits four mutually incompatible failed-cell corners in each case. For multi-failure repairs, we prove a general inclusion--exclusion identity for overlap inside the failed region; hence the formula remains exact for arbitrary higher-order dense cores. When a canonical repair instance is certified to have maximum multiplicity three, the identity reduces to the compact correction $\Omega_{\rm extra}=P_2-A-C_3$. A ground-truth audit over 7,494 Gaussian cases recomputes coverage from lattice geometry, verifies all exact formulas, and records reproducible multiplicity witnesses.
\end{abstract}


\section{Introduction}

Resource placement in interconnection networks asks how special service nodes, such as storage modules, I/O units, controllers, or replicated software services, should be distributed so that all processors have bounded access cost. A particularly strong placement is a \emph{perfect $t$-placement}: every vertex is within distance $t$ of exactly one resource. Equivalently, the radius-$t$ balls centered at the resource nodes partition the network.

For Gaussian interconnection networks, the fault-free perfect-placement problem is already essentially complete. Gaussian networks were modeled using the Gaussian integers in \cite{MartinezGaussian2008}, and Mary and Bose later classified the canonical perfect placements in Gaussian and Eisenstein--Jacobi networks \cite{FlahiveBose2013}. This paper therefore does not attempt to rediscover the perfect placement. Instead, it starts after deployment: assume a classified perfect Gaussian placement is already active, and ask what happens after one or more resource nodes fail.

Fault tolerance in graph-based networks has also been studied through complementary robustness parameters, including fault-tolerant metric dimension in multistage interconnection networks \cite{PrabhuKlavzar2022}, conditional and extra diagnosability of interconnection networks \cite{ChengMaoQiuShen2022}, and fault-tolerant variants of domination such as power domination \cite{GirishSomasundaram2024}. The broader domination literature also contains surveys and taxonomies of perfect, efficient, and independent domination terminology \cite{Klostermeyer2015,GoddardHenning2013}. In parallel, Gaussian-network research has studied structural and communication aspects of dense Gaussian topologies, including routing and broadcasting in on-chip multiprocessor networks \cite{MartinezDenseGaussian2006}. The present paper is different in focus: the fault-free resource set is already a perfect dominating set, and the objective is not to design a globally fault-tolerant set in advance but to determine the exact local replacement cost after resource failures.

The operational goal is local repair. When a resource fails, the vertices that were served by that resource become uncovered. Rather than globally recomputing a new perfect placement, the network may activate nearby replacement resources. This creates a new optimization problem: restore $t$-domination locally, minimize the number of replacement resources, and quantify the unavoidable overlap inside the failed cells.

The contributions of this paper are as follows.
\begin{itemize}
    \item We formulate local fault repair for perfect resource placements in dense Gaussian networks as a post-deployment domination-recovery problem.
    \item We show that the two Mary--Bose dense Gaussian generator families are isometric for local repair, so it suffices to prove the results for $t+(t+1)i$.
    \item We prove locality: after one resource fails, the only uncovered vertices are the vertices of the failed Lee ball, and useful replacement candidates lie within distance $2t$.
    \item We prove the exact one-failure replacement numbers $\rho_G(1)=3$ and $\rho_G(t)=2$ for all $t\ge2$.
    \item We prove the exact one-failure minimum-overlap formula $\Omega_G(t)=t+1$ among minimum-size repairs.
    \item We prove exact two-fault additivity for $t\ge2$: every pair of failed Gaussian resource cells requires exactly four local replacements, and four always suffice.
    \item We prove a general multi-failure overlap identity using exact inclusion--exclusion over replacement multiplicities in the failed region.
    \item We give a certificate-based triple-core specialization for canonical local repairs and include an independent ground-truth audit of 7,494 Gaussian cases, recomputing coverage and overlap directly from lattice geometry.
\end{itemize}

Table~\ref{tab:main-results} summarizes the main formal results.

\begin{table}[H]
\centering
\caption{Main results and proof locations.}
\label{tab:main-results}
\begin{adjustbox}{max width=\textwidth}
\begin{tabular}{p{0.22\textwidth}p{0.33\textwidth}p{0.35\textwidth}}
\toprule
Setting & Result & Proof method\\
\midrule
One failed resource & $\rho_G(1)=3$ and $\rho_G(t)=2$ for $t\ge2$ & Lee-ball extremal points and explicit diagonal/anti-diagonal repairs\\
One failed resource, minimum overlap & $\Omega_G(t)=t+1$ for $t\ge2$ & Rotated-coordinate slice lower bound and tight interface construction\\
Two failed resources & $\rho_G^{(2)}(t,\Delta)=4$ for every pair and every $t\ge2$ & Rotated-square displacement reduction and four incompatible corners\\
General multi-failure overlap & $O(R)=\sum_{j\ge2}(-1)^jP_j(R)$ & Exact vertexwise inclusion--exclusion over replacement multiplicities\\
Certified triple-core instances & $\Omega_{\rm extra}=P_2-A-C_3$ when $\max\mu_R\le3$ & Conditional corollary plus multiplicity certificate\\
\bottomrule
\end{tabular}
\end{adjustbox}
\end{table}

\section{Gaussian Preliminaries}

The Gaussian integer grid is $\Z[i]=\{x+yi:x,y\in\Z\}$ with adjacency defined by differences $\pm1$ and $\pm i$. In coordinate form we write vertices as integer pairs $(x,y)$, and the graph distance from the origin is the Lee distance
\begin{equation}
    d_G((0,0),(x,y))=|x|+|y|.
\end{equation}
The radius-$t$ ball centered at $c$ is
\begin{equation}
    B_t(c)=\{u:d_G(u,c)\le t\},
\end{equation}
and
\begin{equation}
    |B_t(0)|=2t^2+2t+1.
\end{equation}

\begin{definition}[Perfect $t$-placement]
A resource set $S$ in a finite Gaussian network is a perfect $t$-placement if every vertex $u$ is covered by exactly one resource ball:
\begin{equation}
    |\{s\in S:d_G(u,s)\le t\}|=1.
\end{equation}
\end{definition}

\begin{theorem}[Gaussian case of Mary--Bose classification \cite{FlahiveBose2013}]
Let $G_\alpha$ be a Gaussian network generated by $\alpha$. Then $G_\alpha$ has a perfect $t$-dominating set exactly in the canonical divisibility cases, generated by associates of
\begin{equation}
    t+(t+1)i.
\end{equation}
In these cases the resource set is a translate of the corresponding ideal placement.
\end{theorem}

\begin{proposition}[Companion-generator symmetry]\label{prop:companion-symmetry}
The two dense Gaussian generator families
\[
    \alpha_1=t+(t+1)i
    \qquad\text{and}\qquad
    \alpha_2=(t+1)+ti
\]
are equivalent for the local repair problem. The map \(\Phi(z)=i\overline{z}\) is an isometry of the Gaussian grid that sends \(\alpha_1\) to \(\alpha_2\). Consequently, all local replacement numbers and all overlap identities proved for \(t+(t+1)i\) hold unchanged for the companion generator \((t+1)+ti\).
\end{proposition}

\begin{proof}
For \(z=x+yi\), \(\Phi(z)=i(x-yi)=y+xi\), so \(\Phi\) swaps the two integer coordinates. Therefore
\[
    d_G(\Phi(z),\Phi(w))=|y_z-y_w|+|x_z-x_w|=d_G(z,w).
\]
Also
\[
    \Phi(t+(t+1)i)=i(t-(t+1)i)=(t+1)+ti.
\]
Thus \(\Phi\) maps radius-\(t\) Lee balls to radius-\(t\) Lee balls, maps perfect-placement cells in the first dense family to the corresponding cells in the companion family, and preserves coverage multiplicities. Since the repair model is defined only through these metric and incidence relations, the two families have identical local repair behavior.
\end{proof}

By Proposition~\ref{prop:companion-symmetry}, it suffices to present the proofs for one representative of the two dense Gaussian generator families. Throughout the paper we take the quotient induced by
\begin{equation}
    \alpha=t+(t+1)i,
\end{equation}
so that the network size is
\begin{equation}
    N=t^2+(t+1)^2=2t^2+2t+1.
\end{equation}

\section{Fault-Recovery Model}

Let $S$ be a perfect $t$-placement and let $F\subseteq S$ be the set of failed resource nodes. The active resources are $S\setminus F$. A local repair activates a replacement set $R$ of non-failed nodes. The repaired resource set is
\begin{equation}
    S'=(S\setminus F)\cup R.
\end{equation}
For one failed resource $r$, a repair is valid if
\begin{equation}
    B_t(r)\subseteq \bigcup_{x\in R}B_t(x).
\end{equation}
For several failed resources $F=\{r_1,\ldots,r_q\}$, the failed region is
\begin{equation}
    U(F)=\bigcup_{j=1}^q B_t(r_j),
\end{equation}
and a repair is valid if $U(F)$ is covered by the replacement balls.

For a one-failure repair, define $\rho_G(t)$ to be the minimum number of replacements needed to restore domination. Among minimum-size repairs, define $\Omega_G(t)$ to be the minimum number of vertices in the failed cell covered by at least two replacement balls.

\section{Locality Theorems}

\begin{theorem}[Failure-cell locality]\label{thm:failure-cell-locality}
Let $S$ be a perfect $t$-placement and let $r\in S$ fail. Then the only vertices that can become uncovered are the vertices in $B_t(r)$.
\end{theorem}

\begin{proof}
Since $S$ is a perfect $t$-placement, every vertex $u$ is covered by exactly one resource in $S$. If $u\notin B_t(r)$, then $r$ was not the resource covering $u$. Thus the unique resource that covered $u$ remains active after $r$ is removed. Only vertices originally covered by $r$, namely the vertices in $B_t(r)$, can become uncovered.
\end{proof}

\begin{theorem}[Candidate locality]\label{thm:candidate-locality}
If a replacement node $x$ satisfies $d_G(x,r)>2t$, then $B_t(x)\cap B_t(r)=\emptyset$. Hence $x$ cannot contribute to repairing the failed cell $B_t(r)$.
\end{theorem}

\begin{proof}
If $u\in B_t(x)\cap B_t(r)$, then by the triangle inequality,
\begin{equation}
    d_G(x,r)\le d_G(x,u)+d_G(u,r)\le2t,
\end{equation}
contradicting $d_G(x,r)>2t$.
\end{proof}

\section{One-Fault Exact Local Repair}\label{sec:one-fault}

By translation, assume that the failed resource is the origin. The failed cell is $B_t(0)$.

\begin{lemma}[One replacement is impossible]\label{lem:one-replacement-impossible-g}
No single replacement node $x\ne0$ can cover the entire failed cell $B_t(0)$.
\end{lemma}

\begin{proof}
If $B_t(0)\subseteq B_t(x)$, then the two finite balls have the same size, hence $B_t(0)=B_t(x)$. The four extreme points $(t,0)$, $(-t,0)$, $(0,t)$, and $(0,-t)$ determine the center of a Lee ball uniquely, forcing $x=0$. This contradicts that the failed resource itself is unavailable as a replacement.
\end{proof}

The constructive proof uses rotated coordinates
\begin{equation}\label{eq:uv-transform}
    u=x+y,\qquad v=x-y.
\end{equation}
The Lee ball becomes the parity sublattice of an axis-aligned square:
\begin{equation}\label{eq:lee-ball-square}
    B_t(0)=\{(x,y): |u|\le t, |v|\le t, u\equiv v\pmod2\}.
\end{equation}

\begin{figure}[H]
\centering
\safeincludegraphics[width=0.65\linewidth]{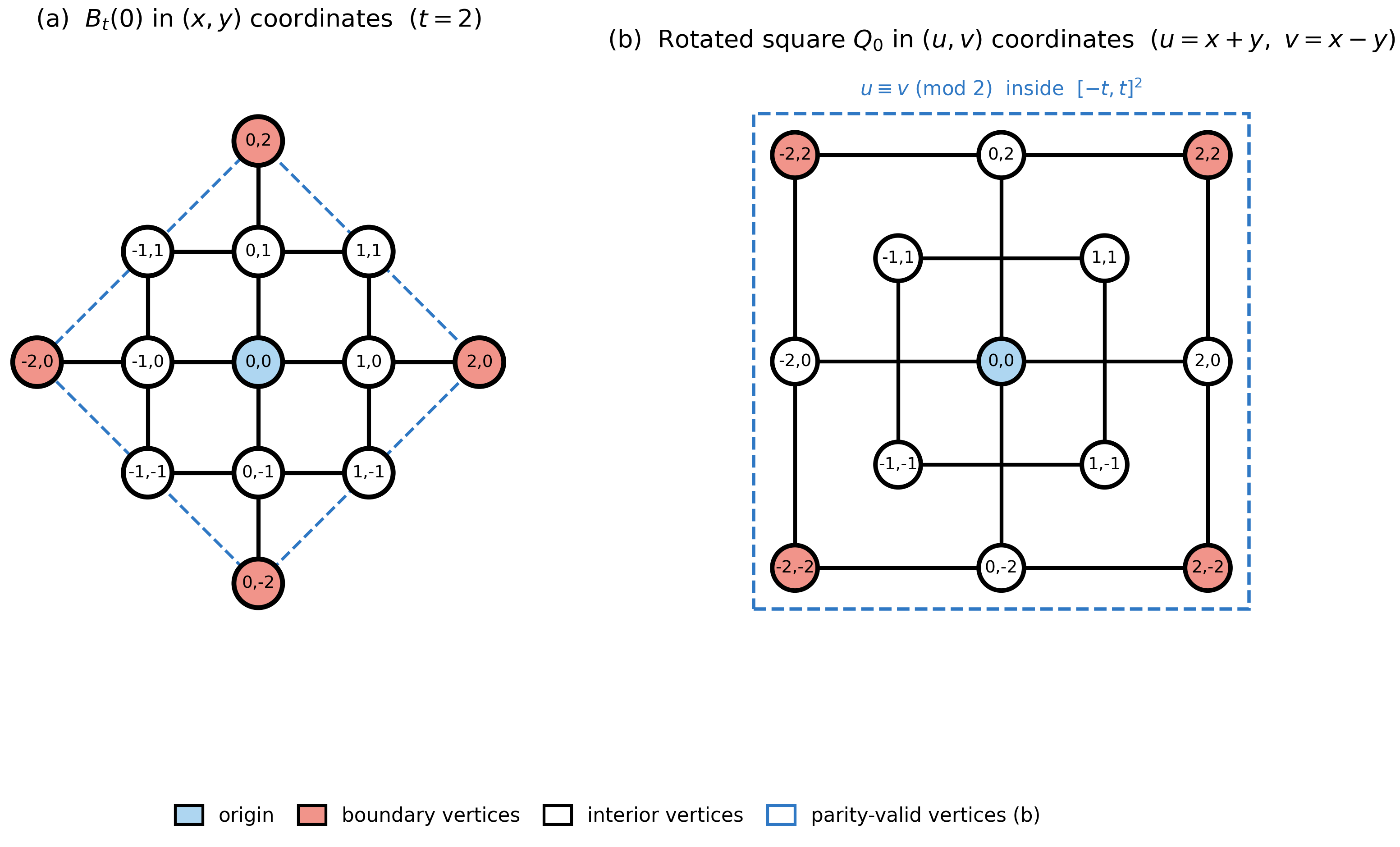}
\caption{Rotated-coordinate representation. The Lee ball $B_t(0)$ in $(x,y)$ coordinates becomes a parity-constrained square in $(u,v)$ coordinates. This geometric picture is the basis for all one-fault and two-fault lower bounds.}
\label{fig:uv-transform}
\end{figure}

\begin{lemma}[$t=1$ exception]\label{lem:t1-exception}
For $t=1$, the failed cell $B_1(0)$ has five vertices, two replacement balls cannot cover it, and three can. Hence $\rho_G(1)=3$.
\end{lemma}

\begin{proof}
The failed cell is
\begin{equation}
B_1(0)=\{(0,0),(1,0),(-1,0),(0,1),(0,-1)\}.
\end{equation}
A radius-one ball centered at a nonzero node covers the origin only when the center is one of the four neighbors of the origin, and then it covers exactly the origin and its own center inside $B_1(0)$. A radius-one ball centered at a diagonal node covers two boundary vertices but not the origin; a distance-two axial center covers one boundary vertex. Therefore two candidates cannot cover all five vertices: either the origin is uncovered, or at most three of the four boundary vertices are covered. Three replacements suffice, for example
\begin{equation}
R=\{(-1,-1),(0,1),(2,0)\}.
\end{equation}
These cover the two negative boundary vertices, the origin and north boundary vertex, and the east boundary vertex, respectively.
\end{proof}

\begin{figure}[H]
\centering
\safeincludegraphics[width=0.65\linewidth]{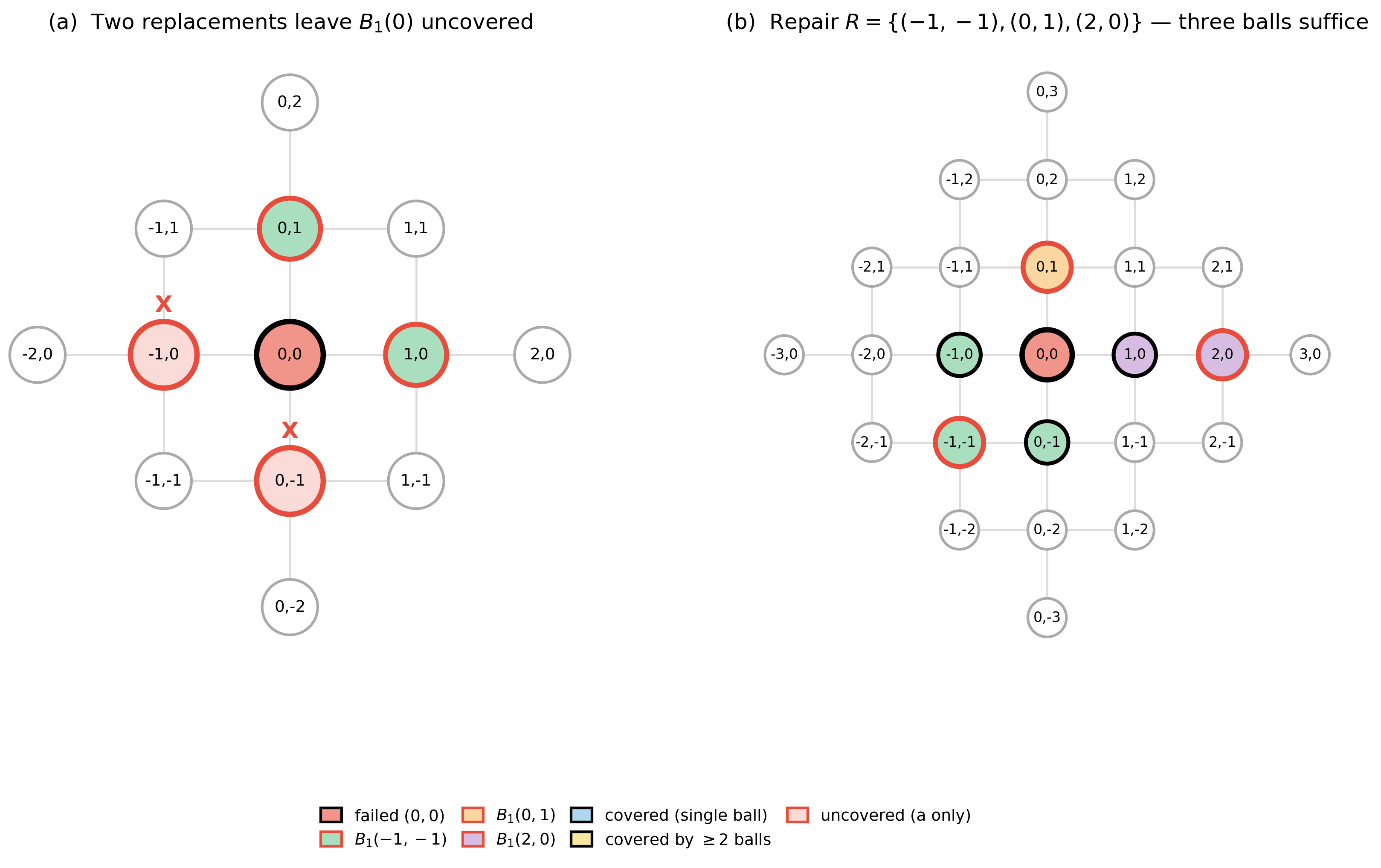}
\caption{Special case $t=1$. Two replacement balls cannot cover all five vertices of the failed cell, while the repair set $\{(-1,-1),(0,1),(2,0)\}$ covers the cell.}
\label{fig:t1-special}
\end{figure}

\begin{theorem}[Two-replacement construction]\label{thm:g-two-replacement}
For $t\ge2$ and every $1\le a\le t$, the diagonal pair
\begin{equation}
    R_a^+=\{(-a,-a),(t-a,t-a)\}
\end{equation}
covers $B_t(0)$. The anti-diagonal pair
\begin{equation}
    R_a^-=\{(-a,a),(t-a,-t+a)\}
\end{equation}
also covers $B_t(0)$.
\end{theorem}

\begin{proof}
It is enough to prove the diagonal case. In the coordinates \eqref{eq:uv-transform}, the ball centered at $(-a,-a)$ contains a failed-cell point exactly when
\begin{equation}
    \max\{|u+2a|,|v|\}\le t.
\end{equation}
Inside $B_t(0)$, the inequality $|v|\le t$ already holds, so this reduces to $u\le t-2a$. Similarly, the ball centered at $(t-a,t-a)$ contains a failed-cell point exactly when $u\ge t-2a$. Every failed-cell point satisfies one of these two inequalities, so the two balls cover $B_t(0)$. Reflection gives the anti-diagonal construction.
\end{proof}

\begin{example}[Worked $t=2$ repair]\label{ex:t2-repair}
For $t=2$, choose $a=1$ in Theorem~\ref{thm:g-two-replacement}. The repair pair is
\[
    R=\{(-1,-1),(1,1)\}.
\]
In rotated coordinates, the two centers are $(-2,0)$ and $(2,0)$. For any point of $B_2(0)$, the first replacement ball covers exactly the side with $u\le0$, while the second covers exactly the side with $u\ge0$. Their common interface is the parity-valid slice $u=0$, containing the three vertices
\[
    (-1,1),\quad (0,0),\quad (1,-1).
\]
Thus the concrete $t=2$ instance already shows both the two-replacement coverage mechanism and the sharp overlap value $t+1=3$.
\end{example}

\begin{figure}[H]
\centering
\safeincludegraphics[width=0.65\linewidth]{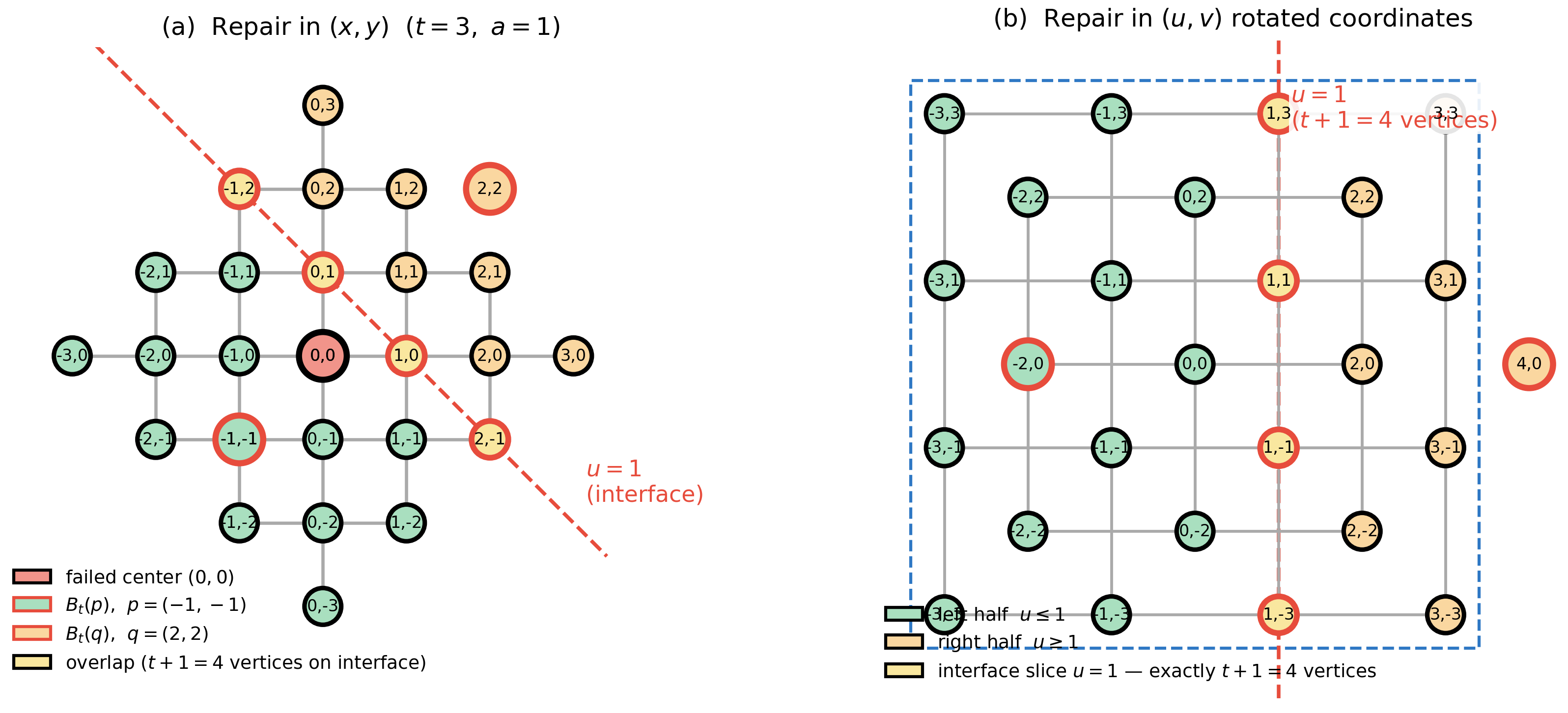}
\caption{Diagonal one-fault repair. The two replacement balls split the failed square at an interface slice; the two balls overlap on exactly $t+1$ parity-valid vertices.}
\label{fig:diagonal-repair}
\end{figure}

\begin{theorem}[Exact one-fault Gaussian repair number]\label{thm:rho-g}
For dense Gaussian networks with canonical perfect $t$-placements,
\begin{equation}
\rho_G(t)=
\begin{cases}
3, & t=1,\\
2, & t\ge2.
\end{cases}
\end{equation}
\end{theorem}

\begin{proof}
For $t\ge2$, Theorem~\ref{thm:g-two-replacement} gives a two-replacement repair, while Lemma~\ref{lem:one-replacement-impossible-g} rules out one replacement. For $t=1$, Lemma~\ref{lem:t1-exception} gives necessity and sufficiency of three replacements.
\end{proof}

\section{Minimum Overlap for One Fault}\label{sec:minimum-overlap}

Let a two-replacement repair be $R=\{p,q\}$. The overlap inside the failed cell is
\begin{equation}
    |B_t(p)\cap B_t(q)\cap B_t(0)|.
\end{equation}

\begin{lemma}[Corner pairing in a two-square cover]\label{lem:corner-pairing}
Let two radius-$t$ Gaussian balls centered away from the origin cover $B_t(0)$. In rotated coordinates, their two parity-constrained squares cover the four corners of $Q_0$. Then, up to interchanging $u$ and $v$ and reflecting coordinates, one square covers the two left corners of $Q_0$ and the other covers the two right corners of $Q_0$.
\end{lemma}

\begin{proof}
A radius-$t$ square distinct from $Q_0$ cannot contain two opposite corners of $Q_0$. Indeed, the only axis-aligned square of side length $2t$ containing both $(-t,-t)$ and $(t,t)$ is $Q_0$ itself, and the same holds for the other opposite pair $(-t,t)$ and $(t,-t)$. Since the failed resource at the origin is unavailable, neither replacement square can cover an opposite-corner pair.

Thus each replacement square can cover at most two adjacent corners of $Q_0$. Because the two squares together must cover all four corners, their covered corner pairs must be complementary adjacent pairs. Up to symmetry, these are the left pair and the right pair. The other possibility, bottom and top, is obtained by interchanging $u$ and $v$.
\end{proof}

\begin{lemma}[Forced interface slice]\label{lem:forced-interface-slice}
Let two replacement balls cover $B_t(0)$ for $t\ge2$. Then their overlap inside $B_t(0)$ contains a complete parity-valid slice of $B_t(0)$ containing exactly $t+1$ vertices.
\end{lemma}

\begin{proof}
By Lemma~\ref{lem:corner-pairing}, it is enough to consider the case in which one replacement square covers the left corners of $Q_0$ and the other covers the right corners. Write their rotated centers as $(U_1,V_1)$ and $(U_2,V_2)$. To contain both left corners $(-t,-t)$ and $(-t,t)$, the first square must have $V_1=0$ and $U_1\le0$. To contain both right corners $(t,-t)$ and $(t,t)$, the second square must have $V_2=0$ and $U_2\ge0$.

Their projections on the $u$-axis are the intervals $[U_1-t,U_1+t]$ and $[U_2-t,U_2+t]$. Since the two balls cover every vertex of $Q_0$, these intervals cover every integer slice $-t\le u\le t$; hence
\begin{equation}
    U_2-U_1\le 2t.
\end{equation}
Therefore the overlap of the two $u$-intervals contains the integer interval
\begin{equation}
    [U_2-t,\,U_1+t].
\end{equation}
This interval is nonempty. Because $V_1=V_2=0$, both squares contain the full $v$-range $[-t,t]$ on every slice whose $u$-coordinate lies in this interval.

The centers of Gaussian vertices satisfy $U_j\equiv V_j\pmod2$ because $U_j=x_j+y_j$ and $V_j=x_j-y_j$ have the same parity for every integer Gaussian vertex $(x_j,y_j)$. Since $V_1=V_2=0$ and $0$ is even, the congruences $U_1\equiv V_1$ and $U_2\equiv V_2\pmod2$ imply that both $U_1$ and $U_2$ are even. If the interval has more than one integer value, it contains a value $c\equiv t\pmod2$. If the interval has exactly one value, then $U_2-U_1=2t$ and that value is $c=U_1+t=U_2-t$, again satisfying $c\equiv t\pmod2$. Thus the overlap contains a full slice $u=c$ with $c\equiv t\pmod2$.

On such a slice, the parity-valid values of $v$ in $[-t,t]$ are
\begin{equation}
    -t,-t+2,\ldots,t,
\end{equation}
so the slice contains exactly $t+1$ Gaussian vertices. Hence the two replacement balls overlap inside $B_t(0)$ on at least $t+1$ vertices.
\end{proof}

\begin{theorem}[Exact Gaussian minimum-overlap formula]\label{thm:omega-g}
For $t\ge2$, among all minimum-size one-failure repairs,
\begin{equation}
    \Omega_G(t)=t+1.
\end{equation}
\end{theorem}

\begin{proof}
By Theorem~\ref{thm:rho-g}, a minimum repair has two replacements. Lemma~\ref{lem:forced-interface-slice} shows that any two-replacement repair has at least $t+1$ overlapped vertices in the failed cell. The construction in Theorem~\ref{thm:g-two-replacement} is tight: its two balls meet inside $B_t(0)$ exactly on the interface slice $u=t-2a$, which has $t+1$ parity-valid vertices. Therefore the minimum overlap is exactly $t+1$.
\end{proof}

\section{Exact Two-Fault Additive Repair}\label{sec:two-fault}

For two failed resources $r_1,r_2$, let
\begin{equation}
    U_2(r_1,r_2)=B_t(r_1)\cup B_t(r_2).
\end{equation}

\begin{figure}[H]
\centering
\safeincludegraphics[width=0.55\linewidth]{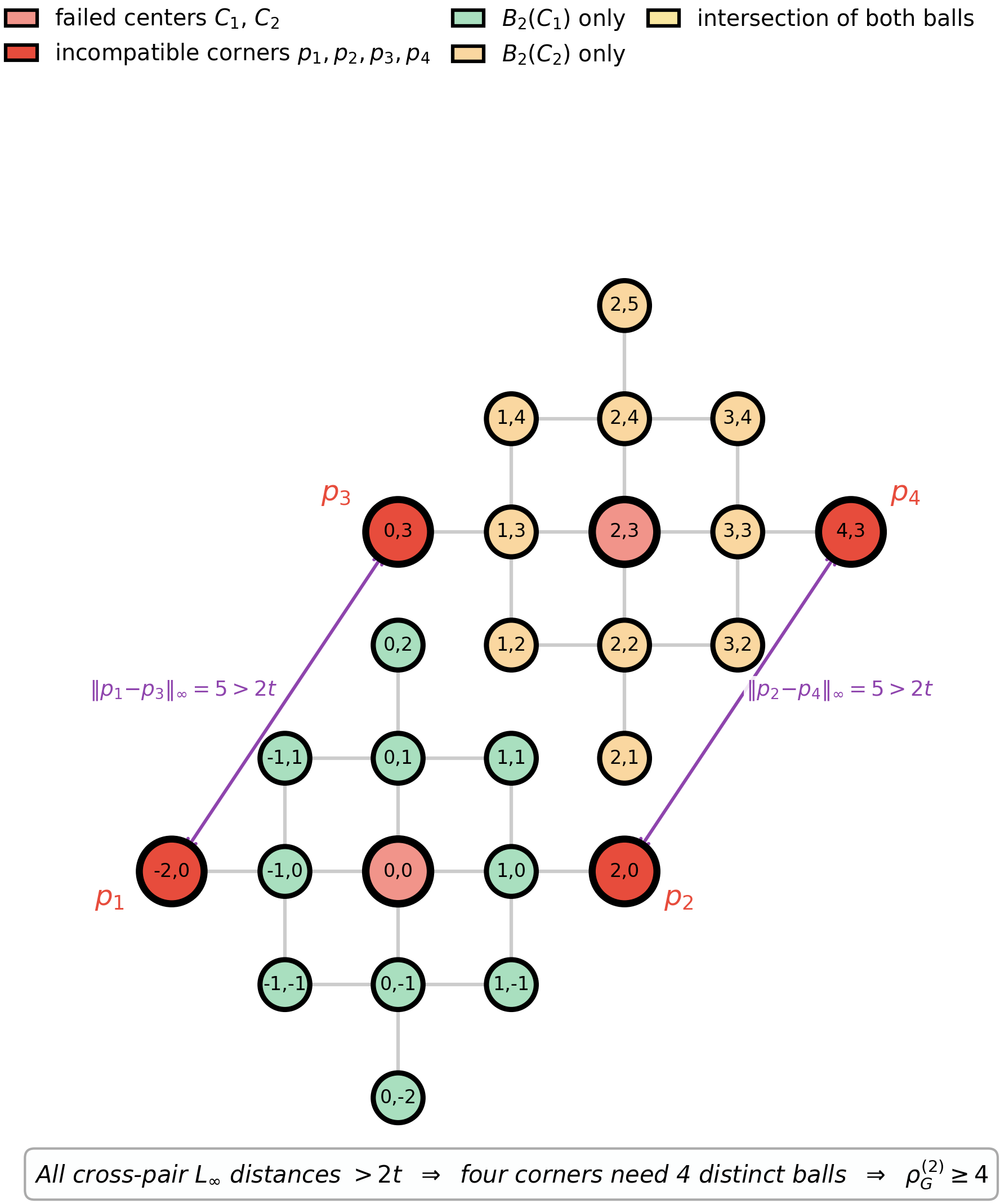}
\caption{Two-fault side-neighbor obstruction. Four marked corners $p_1,\ldots,p_4$ force four distinct replacement balls in the side-neighbor case of Lemma~\ref{lem:side-adjacent-four-corners}.}
\label{fig:two-fault-side}
\end{figure}

\subsection{Rotated-square representation}

We use the rotated coordinates
\begin{equation}
    u=x+y,\qquad v=x-y.
\end{equation}
In these coordinates the Lee ball centered at $(U_0,V_0)$ is the parity-constrained square
\begin{equation}
\begin{aligned}
    Q_t(U_0,V_0)=\{(u,v):{}& |u-U_0|\le t,\ |v-V_0|\le t,\\
    & u\equiv v\pmod 2\}.
\end{aligned}
\end{equation}

In the dense Gaussian perfect placement generated by $t+(t+1)i$, the resource-center lattice has rotated-coordinate generators
\begin{equation}
    e_1=(2t+1,-1),\qquad e_2=(-1,-(2t+1)).
\end{equation}

\begin{theorem}[Two-fault additive construction]\label{thm:two-fault-additive-upper}
For $t\ge2$, any two failed Gaussian resource cells can be repaired using at most four replacement nodes.
\end{theorem}

\begin{proof}
Translate the two-replacement construction of Theorem~\ref{thm:g-two-replacement} to $r_1$ and to $r_2$. The first translated pair covers $B_t(r_1)$ and the second translated pair covers $B_t(r_2)$. Their union therefore covers $U_2(r_1,r_2)$. Hence four replacements always suffice.
\end{proof}

\begin{lemma}[Relevant neighboring displacements]\label{lem:relevant-two-fault-displacements}
Let $t\ge2$. If a single replacement ball can intersect both failed cells, then, up to sign, rotation, and reflection of the rotated square coordinates, the displacement between the two failed resource centers is one of
\begin{equation}
    (2t+1,-1)
    \quad\text{or}\quad
    (2t+2,2t).
\end{equation}
\end{lemma}

\begin{proof}
Let the two failed resource centers have rotated displacement $D=(D_u,D_v)$. A replacement ball of radius $t$ can intersect both failed cells only if the two failed radius-$t$ squares are within $2t$ of one another in the $L_\infty$ metric. Hence
\begin{equation}
    \|D\|_\infty\le 4t.
\end{equation}
Write
\begin{equation}
    D=m(2t+1,-1)+n(-1,-(2t+1)),
\end{equation}
where $m,n\in\mathbb Z$. Put $a=2t+1$. Then
\begin{equation}
    D_u=am-n,\qquad D_v=-m-an.
\end{equation}
Solving gives
\begin{equation}
    m=\frac{aD_u-D_v}{a^2+1},
    \qquad
    n=-\frac{D_u+aD_v}{a^2+1}.
\end{equation}
Since $|D_u|,|D_v|\le4t$,
\begin{equation}
    |aD_u-D_v|\le 4t(a+1)=8t(t+1)<2(a^2+1),
\end{equation}
for $a=2t+1$. Therefore $|m|<2$. Similarly, $|n|<2$. Hence $m,n\in\{-1,0,1\}$.

The nonzero possibilities are $(\pm1,0)$, $(0,\pm1)$, and $(\pm1,\pm1)$. The cases $(\pm1,0)$ and $(0,\pm1)$ are equivalent under the isometry of the rotated-square model that interchanges the $u$- and $v$-axes and then reflects signs as needed. Therefore they give the same side-neighbor displacement type $(2t+1,-1)$ up to the allowed lattice symmetries. The cases $(\pm1,\pm1)$ are likewise equivalent and give the displacement type $(2t+2,2t)$.
\end{proof}

\begin{lemma}[Four incompatible corners for side-neighbor cells]\label{lem:side-adjacent-four-corners}
Let the two failed cells have centers
\begin{equation}
    C_1=(0,0),\qquad C_2=(2t+1,-1)
\end{equation}
in rotated coordinates. Then any repair requires at least four replacement balls.
\end{lemma}

\begin{proof}
Let $Q_1=Q_t(C_1)$ and $Q_2=Q_t(C_2)$. Consider the four failed-cell corner vertices
\begin{align}
    p_1&=(-t,-t),&
    p_2&=(t,t),\nonumber\\
    p_3&=(t+1,-t-1),&
    p_4&=(3t+1,t-1).
\end{align}
The points $p_1,p_2$ are opposite corners of $Q_1$, and $p_3,p_4$ are opposite corners of $Q_2$. All four points satisfy the required parity condition $u\equiv v\pmod2$.

A radius-$t$ square can cover two points only if their $L_\infty$ distance is at most $2t$. Moreover, the only radius-$t$ square containing two opposite corners of the same failed square is the failed square itself. Hence $p_1$ and $p_2$ can be covered together only by the failed center $C_1$, and $p_3$ and $p_4$ can be covered together only by the failed center $C_2$. Both centers are unavailable.

For the cross pairs,
\begin{align}
    \|p_1-p_3\|_\infty &= 2t+1,&
    \|p_1-p_4\|_\infty &= 4t+1,\nonumber\\
    \|p_2-p_3\|_\infty &= 2t+1,&
    \|p_2-p_4\|_\infty &= 2t+1.
\end{align}
Each value is greater than $2t$. Therefore no valid radius-$t$ replacement ball can cover any cross pair. Thus the four vertices $p_1,p_2,p_3,p_4$ must be covered by four distinct replacement balls, and at least four replacements are necessary.
\end{proof}

\begin{lemma}[Four incompatible corners for diagonal-neighbor cells]\label{lem:diagonal-adjacent-four-corners}
Let the two failed cells have centers
\begin{equation}
    C_1=(0,0),\qquad C_2=(2t+2,2t)
\end{equation}
in rotated coordinates. Then any repair requires at least four replacement balls.
\end{lemma}

\begin{proof}
Let $Q_1=Q_t(C_1)$ and $Q_2=Q_t(C_2)$. Consider the four failed-cell corner vertices
\begin{align}
    q_1&=(-t,t),&
    q_2&=(t,-t),\nonumber\\
    q_3&=(t+2,3t),&
    q_4&=(3t+2,t).
\end{align}
The points $q_1,q_2$ are opposite corners of $Q_1$, and $q_3,q_4$ are opposite corners of $Q_2$. Again, all four points satisfy $u\equiv v\pmod2$.

The pair $q_1,q_2$ can be covered together only by the failed center $C_1$, and the pair $q_3,q_4$ can be covered together only by the failed center $C_2$. Both centers are unavailable. For the cross pairs,
\begin{align}
    \|q_1-q_3\|_\infty &= 2t+2,&
    \|q_1-q_4\|_\infty &= 4t+2,\nonumber\\
    \|q_2-q_3\|_\infty &= 4t,&
    \|q_2-q_4\|_\infty &= 2t+2.
\end{align}
Each value is greater than $2t$ for $t\ge2$. Therefore no valid replacement ball can cover a cross pair. The four vertices $q_1,q_2,q_3,q_4$ require four distinct replacement balls, and every repair has size at least four.
\end{proof}

\begin{lemma}[Far failed cells require independent repairs]\label{lem:far-cells-independent}
If two failed resource centers have rotated displacement $D$ with
\begin{equation}
    \|D\|_\infty>4t,
\end{equation}
then any replacement ball intersects at most one failed cell. Consequently, at least four replacement balls are required.
\end{lemma}

\begin{proof}
If a radius-$t$ replacement ball intersected both failed cells, then the two failed radius-$t$ squares would be within $2t$ of each other in the $L_\infty$ metric. Since each failed cell itself has radius $t$, this would imply $\|D\|_\infty\le4t$, contradicting the hypothesis. Thus every replacement ball can contribute to at most one failed cell. By Theorem~\ref{thm:rho-g}, one failed Gaussian cell requires at least two replacement balls for $t\ge2$. Therefore two such cells require at least four replacements.
\end{proof}

\begin{theorem}[Exact two-fault Gaussian additivity]\label{thm:two-fault-exact}
For every $t\ge2$ and every pair of failed resource centers in the dense Gaussian perfect placement,
\begin{equation}
    \rho_G^{(2)}(t,\Delta)=4.
\end{equation}
\end{theorem}

\begin{proof}
The upper bound $\rho_G^{(2)}(t,\Delta)\le4$ follows from Theorem~\ref{thm:two-fault-additive-upper}. It remains to prove the lower bound.

Let $D$ be the rotated displacement between the two failed resource centers. If $\|D\|_\infty>4t$, Lemma~\ref{lem:far-cells-independent} gives the lower bound four. Otherwise, $\|D\|_\infty\le4t$. By Lemma~\ref{lem:relevant-two-fault-displacements}, up to the symmetries of the Gaussian lattice, $D$ is either $(2t+1,-1)$ or $(2t+2,2t)$. In the first case, Lemma~\ref{lem:side-adjacent-four-corners} gives the lower bound four. In the second case, Lemma~\ref{lem:diagonal-adjacent-four-corners} gives the lower bound four.

Thus every two-fault repair requires at least four replacement balls. Since four always suffice, the exact value is $\rho_G^{(2)}(t,\Delta)=4$.
\end{proof}

\section{Multi-Fault Overlap Accounting}\label{sec:dense-core}

\begin{figure}[H]
\centering
\safeincludegraphics[width=0.65\linewidth]{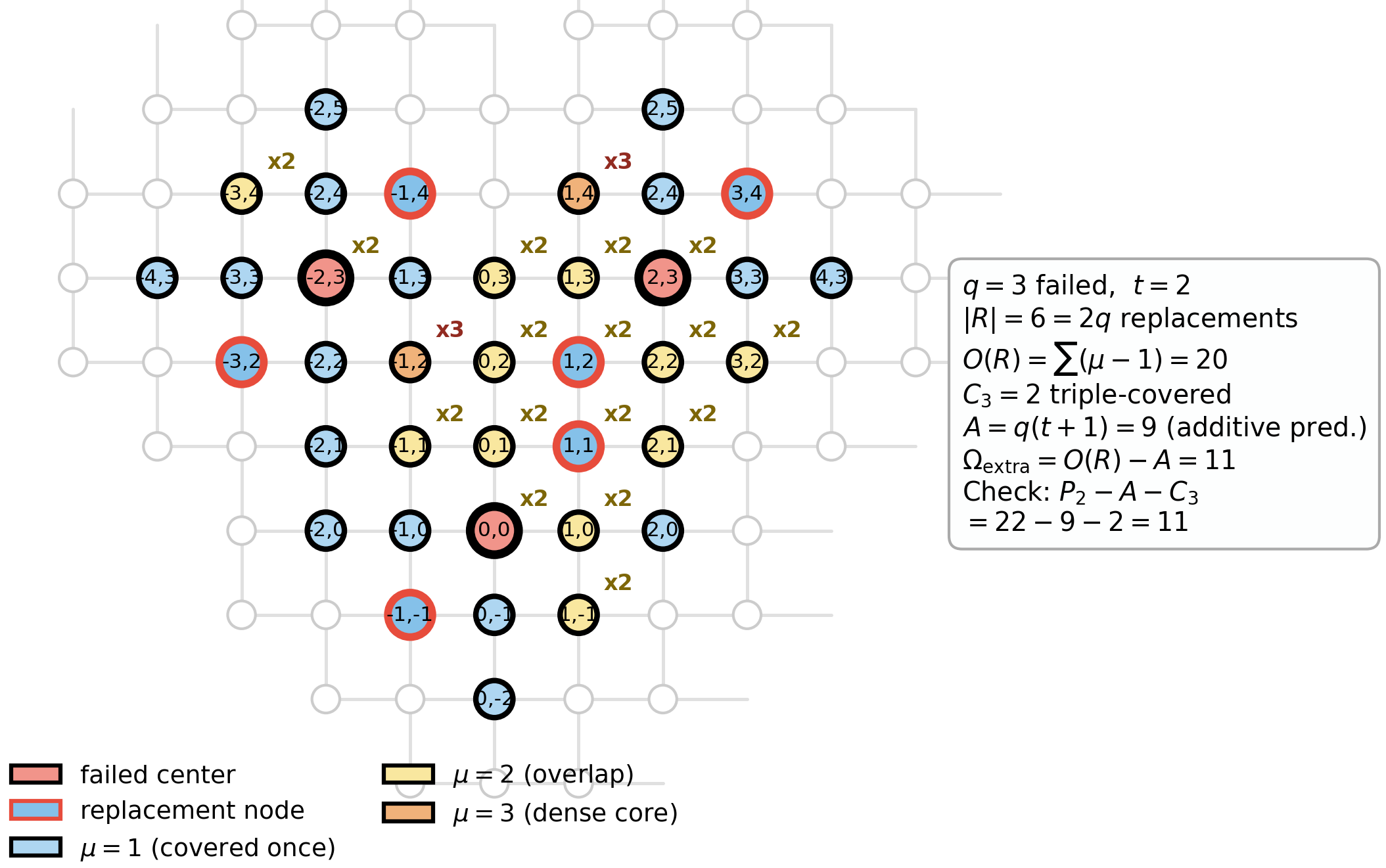}
\caption{Multi-fault overlap and dense cores. For $q=3$ or $q=4$, overlapping replacement squares can create vertices covered by three or more balls; these vertices are exactly what the inclusion--exclusion correction accounts for.}
\label{fig:multi-fault-core}
\end{figure}

Let
\begin{equation}
    U(F)=\bigcup_{r\in F}B_t(r)
\end{equation}
be the failed region and let $R$ be any valid local replacement set. For $z\in U(F)$ define the replacement multiplicity
\begin{equation}
    \mu_R(z)=|\{x\in R:z\in B_t(x)\}|.
\end{equation}
Since $R$ is a valid repair, $\mu_R(z)\ge1$ for all $z\in U(F)$. The true overlap mass inside the failed region is
\begin{equation}\label{eq:true-overlap-mass}
    O(R)=\sum_{z\in U(F)}(\mu_R(z)-1).
\end{equation}
For $j\ge2$, define the $j$-fold coverage mass
\begin{equation}\label{eq:pj-def}
    P_j(R)=\sum_{z\in U(F)} \binom{\mu_R(z)}{j}.
\end{equation}

\begin{lemma}[Vertexwise binomial identity]\label{lem:vertex-binomial}
For every integer $m\ge1$,
\begin{equation}\label{eq:vertex-binomial}
    m-1=\sum_{j=2}^{m}(-1)^j \binom{m}{j}.
\end{equation}
\end{lemma}

\begin{proof}
The binomial identity
\begin{equation}
    \sum_{j=0}^{m}(-1)^j\binom{m}{j}=0
\end{equation}
holds for $m\ge1$. Separating the $j=0$ and $j=1$ terms gives
\begin{equation}
    1-m+\sum_{j=2}^{m}(-1)^j\binom{m}{j}=0,
\end{equation}
which is equivalent to \eqref{eq:vertex-binomial}.
\end{proof}

\begin{theorem}[Exact multi-failure overlap identity]\label{thm:general-overlap-ie}
For any valid local repair $R$ of any failed resource set $F$,
\begin{equation}\label{eq:general-overlap-ie}
    O(R)=\sum_{j=2}^{M}(-1)^j P_j(R),
\end{equation}
where
\begin{equation}
    M=\max_{z\in U(F)}\mu_R(z).
\end{equation}
\end{theorem}

\begin{proof}
Apply Lemma~\ref{lem:vertex-binomial} to $m=\mu_R(z)$ at each failed-region vertex $z\in U(F)$ and sum over all $z$. Exchanging the finite sums gives
\begin{equation}
\begin{aligned}
    O(R)
    &=\sum_{z\in U(F)}(\mu_R(z)-1)\\
    &=\sum_{z\in U(F)}\sum_{j=2}^{\mu_R(z)}(-1)^j \binom{\mu_R(z)}{j}\\
    &=\sum_{j=2}^{M}(-1)^j\sum_{z\in U(F)}\binom{\mu_R(z)}{j}\\
    &=\sum_{j=2}^{M}(-1)^jP_j(R).
\end{aligned}
\end{equation}
\end{proof}

\begin{corollary}[Triple-core specialization]\label{cor:triple-core}
If a repair set $R$ satisfies
\begin{equation}
    \max_{z\in U(F)}\mu_R(z)\le3,
\end{equation}
then
\begin{equation}\label{eq:triple-core}
    O(R)=P_2(R)-P_3(R).
\end{equation}
Moreover, if $C_3(R)$ denotes the number of failed-region vertices covered by exactly three replacement balls, then $P_3(R)=C_3(R)$ and
\begin{equation}\label{eq:triple-core-c3}
    O(R)=P_2(R)-C_3(R).
\end{equation}
\end{corollary}

\begin{proof}
If $M\le3$, Theorem~\ref{thm:general-overlap-ie} has only the $j=2$ and $j=3$ terms. Also, when no vertex has multiplicity above three, $\binom{\mu_R(z)}{3}$ is one exactly at vertices of multiplicity three and zero elsewhere. Hence $P_3(R)=C_3(R)$.
\end{proof}

\begin{theorem}[Certified dense-core correction]\label{thm:dense-core}
For any canonical local repair instance for which the ground-truth certificate verifies
\begin{equation}
    \max_{z\in U(F)}\mu_R(z)\le3,
\end{equation}
the dense-core correction is exactly
\begin{equation}\label{eq:extra-core-final}
    \Omega_{\rm extra}(R)=P_2(R)-A(R)-C_3(R).
\end{equation}
\end{theorem}

\begin{proof}
By Corollary~\ref{cor:triple-core}, $O(R)=P_2(R)-C_3(R)$. Subtracting the additive prediction $A(R)$ gives \eqref{eq:extra-core-final}.
\end{proof}

\begin{remark}[Scope of the certified specialization]\label{rem:scope-certified-specialization}
Theorem~\ref{thm:general-overlap-ie} is unconditional and applies to every local repair. The compact dense-core formula in Theorem~\ref{thm:dense-core} is a specialization for instances whose certificate proves $M\le3$. This distinction prevents the multi-failure result from relying on an unproved universal pattern classification: if a future repair family produces $M>3$, the same theorem still applies with the additional $P_4,P_5,\ldots$ terms.
\end{remark}

\section{Validation and Audit}\label{sec:validation}

Table~\ref{tab:g-one-validation} verifies the exact one-failure construction and minimum-overlap values for representative radii. Table~\ref{tab:g-dense-audit} verifies the multiplicity certificate required to specialize the general inclusion--exclusion identity to the triple-core formula. Every reported quantity is recomputed from raw lattice geometry.

\begin{table}[H]
\centering
\caption{Exact one-failure validation for representative radii.}
\label{tab:g-one-validation}
\begin{tabular}{ccccc}
\toprule
$t$ & $N$ & candidates & min $K$ & best overlap\\
\midrule
1 & 5 & 12 & 3 & 0\\
2 & 13 & 40 & 2 & 3\\
5 & 61 & 220 & 2 & 6\\
10 & 221 & 840 & 2 & 11\\
15 & 481 & 1860 & 2 & 16\\
20 & 841 & 3280 & 2 & 21\\
\bottomrule
\end{tabular}
\end{table}

\begin{table}[H]
\centering
\caption{Dense-core audit summary. All quantities were recomputed from lattice geometry.}
\label{tab:g-dense-audit}
\begin{adjustbox}{max width=\textwidth}
\begin{tabular}{cccccccc}
\toprule
Faults $q$ & Cases & Extra cases & Max extra & Max cover & $C_3$ max & Residual max & $\ge4$-cover cases\\
\midrule
1 & 6    & 0   & 0  & 2 & 0 & 0 & 0\\
2 & 288  & 0   & 0  & 2 & 0 & 0 & 0\\
3 & 1800 & 56  & 15 & 3 & 1 & 0 & 0\\
4 & 1800 & 182 & 24 & 3 & 2 & 0 & 0\\
5 & 1800 & 382 & 58 & 3 & 8 & 0 & 0\\
6 & 1800 & 573 & 56 & 3 & 8 & 0 & 0\\
\midrule
Total & 7494 & 1193 & 58 & 3 & 8 & 0 & 0\\
\bottomrule
\end{tabular}
\end{adjustbox}
\end{table}

\section{Discussion}

The Gaussian-only formulation gives a focused post-deployment repair problem rather than another classification of perfect placements. The companion dense generator $(t+1)+ti$ is not omitted as a separate case: Proposition~\ref{prop:companion-symmetry} shows that it is the image of $t+(t+1)i$ under a Gaussian-grid isometry, so the local repair results transfer unchanged. The one-fault replacement and overlap formulas are exact, and the two-fault repair number is exact for all $t\ge2$ by the rotated-square corner-incompatibility argument.

The multi-failure result has a different role. It is not needed to prove the one- or two-fault replacement numbers; rather, it gives exact overlap accounting when several local repairs interact. The final theorem is a general inclusion--exclusion identity over replacement multiplicities, so it remains correct even if a non-canonical repair produces fourth- or higher-order overlaps. The simpler dense-core expression $P_2-A-C_3$ is used only when the accompanying certificate verifies maximum multiplicity three.

\section{Conclusion}

This paper introduced local fault repair for perfect resource placements in dense Gaussian networks. The two dense Gaussian generator families in the Mary--Bose classification are isometric for this problem, so the proofs for $t+(t+1)i$ also cover the companion generator $(t+1)+ti$. After a resource node fails, the uncovered region is exactly one Lee ball, and all useful replacements lie in a radius-$2t$ neighborhood. The exact one-fault repair number is $3$ for $t=1$ and $2$ for $t\ge2$, and the minimum overlap among minimum-size repairs is $t+1$. For two faults, four local replacements are always sufficient and, for $t\ge2$, always necessary. The lower bound follows from the rotated-square representation of Gaussian Lee balls and a finite reduction to two neighboring displacement types. For multi-failure repairs, the overlap inside the failed region is given by an exact inclusion--exclusion identity over replacement multiplicities. When the canonical repair certificate verifies maximum multiplicity three, this reduces to the dense-core correction $P_2-A-C_3$.

\section*{Acknowledgment}
The author thanks the Department of Computer Science, Faculty of Science, Kuwait University, for its support and research environment. This work did not receive a specific grant from any funding agency in the public, commercial, or not-for-profit sectors.

\end{document}